\documentclass{llncs}
\usepackage{amsmath}

\usepackage{amsfonts}

\title{A discrepancy lower bound for information complexity
}
\author{ 
Mark Braverman\thanks{Princeton University and the University of Toronto, 
\texttt{mbraverm@cs.princeton.edu}. Partially supported by an NSERC Discovery Grant, an Alfred P. Sloan Fellowship, and an NSF CAREER award.}
\institute{} \and
Omri Weinstein\thanks{Princeton University, 
\texttt{oweinste@cs.princeton.edu}.}
}

\newcommand{\getsr}{\in_{R}}
\newcommand{\ring}[1]{\mathbb{#1}}

\newcommand{\Z}{\ring{Z}}

\newcommand{\CProt}[1]{\mathsf{CC}(#1)}

\newcommand{\mathbbb}[1]{{\bf #1}}

\newcommand{\Div}[2]{\mathbbb{D} \left (  #1 ||#2 \right) }

\newcommand{\eqdef}{\stackrel{def}{=}}
\newcommand{\cZ}{\mathcal{Z}}

\newcommand{\ignore}[1]{{}}

\renewcommand{\P}{{\mathbf P}}
\newcommand{\E}{{\mathbf E}}
\newcommand{\cA}{{\cal A}}
\newcommand{\cU}{{\cal U}}
\newcommand{\cW}{{\cal W}}

\newcommand{\ra}{\rightarrow}
\newcommand{\de}{\delta}
\newcommand{\Om}{\Omega}
\newcommand{\Disc}{Disc_{\mu}(f)}
\newcommand{\DiscR}{Disc_{\mu}(R,f)}
\newcommand{\eps}{\epsilon}
\newcommand{\suc}{\mathcal{S}}

\newcommand{\ICmu}[3]{{\mathsf{IC}_{#2}(#1,#3)}}
\newcommand{\IC}[2]{{\mathsf{IC}\left(#1,#2\right)}}

\newcommand{\ICprot}[2]{{\mathsf{IC}_{#2}(#1)}}

\newcommand{\algwidth}{0.97\textwidth}

\newcommand{\cX}{{\cal X}}
\newcommand{\cY}{{\cal Y}}
\newcommand{\cL}{{\cal L}}
\newcommand{\cN}{{\cal N}}
\newcommand{\cB}{{\cal B}}
\newcommand{\cE}{{\cal E}}

\newcommand{\G}{\mathcal{G}}

\newcommand{\Ex}[2]{\mathop{\mathbbb{E}}\displaylimits_{#1}\left
[ #2 \right ]}

\newcommand{\ve}{\varepsilon}
\newcommand{\al}{\alpha}
\newcommand{\be}{\beta}

\begin{document}


\maketitle
\thispagestyle{empty}
\begin{abstract}
This paper provides the first general technique for proving information lower bounds on two-party 
unbounded-rounds communication problems. We show that the discrepancy lower bound, which 
applies to randomized communication complexity, also applies to information complexity. More 
precisely, if the discrepancy of a two-party function $f$ with respect to a distribution $\mu$ is 
$Disc_\mu f$, then any two party randomized protocol computing $f$ must reveal at least 
$\Omega(\log (1/Disc_\mu f))$ bits of information to the participants. As a corollary, we obtain that 
any two-party protocol for computing a random function on $\{0,1\}^n\times\{0,1\}^n$ must reveal 
$\Omega(n)$ bits of information to the participants. 

In addition, we prove that the discrepancy of the Greater-Than function is $\Omega(1/\sqrt{n})$, which provides an alternative proof to the recent proof of Viola \cite{Viola11} of the $\Omega(\log n)$ lower bound 
on the communication complexity  of this well-studied function and, combined with our main result, proves the tight $\Omega(\log n)$ lower bound on its information complexity.

The proof of our main result develops a new simulation procedure that may be of an independent interest. In a very recent breakthrough work of Kerenidis et al. \cite{kerenidis2012lower}, this simulation procedure was a building block towards a proof that almost all known 
lower bound techniques for communication complexity (and not just discrepancy) apply to information complexity.
\end{abstract}


\begin{section}{Introduction}
The main objective of this paper is to expand the available techniques for proving information complexity lower bounds 
for communication problems. Let $f:\cX\times \cY \ra\{0,1\}$ be a function, and $\mu$ be a distribution on $\cX\times\cY$. 
Informally, the information complexity of $f$ is the least amount of {\em information} that Alice and Bob need to exchange 
on average to compute $f(x,y)$ using a randomized communication protocol if initially $x$ is given to Alice, $y$ is given 
to Bob, and $(x,y)\sim \mu$. Note that information here is measured in the Shannon sense, and  the amount of 
information may be much smaller than the number of bits exchanged. Thus the randomized 
communication complexity of $f$ is an upper bound on its information complexity, but may not be a lower bound. 

Within the context of communication complexity, information complexity has first been introduced in the context of direct sum theorems for randomized communication 
complexity \cite{ChakrabartiSWY01,BaryossefJKS04,BarakBCR10}. These techniques are also being used in the 
related direction of direct product theorems \\ \cite{klauck2004quantum,Lee08adirect,jain2010strong,klauck2010strong}. A direct sum theorem in a computational model 
states that the amount of resources needed to perform $k$ independent tasks is roughly $k$ times the amount of resources $c$ needed for 
computing a single task. A direct product theorem, which is a stronger statement, asserts that any attempt to solve $k$ independent tasks using 
$o(kc)$ resources would result in an exponentially small success probability $2^{-\Omega(k)}$.

The direct sum line of work \cite{harsha2007communication,jain2008optimal,BarakBCR10,braverman2011information} 
has eventually led to a tight connection (equality) between amortized communication 
complexity and information complexity. 
Thus proving lower bounds on the communication complexity of $k$ copies of $f$ for a growing $k$ is equivalent to 
proving lower bounds on the information complexity of $f$. 
In particular if $f$ satisfies $IC(f)=\Om(CC(f))$, i.e. that its information 
cost is asymptotically equal to its communication complexity, then a strong direct sum theorem holds for $f$. 
In addition to the intrinsic interest of understanding the amount of information exchange that needs to be 
involved in computing $f$, direct sum theorems motivate 
 the development of techniques for proving lower bounds on the information complexity of functions. 

Another important motivation for seeking lower bounds on the information complexity of functions 
stems from understanding the limits of security in two-party computation. In a celebrated result Ben-Or et al. 
\cite{ben1988completeness} (see also \cite{asharov2011full})
showed how a multi-party computation (with three or more parties) may be carried out in a way that 
reveals no information to the participants except for the computation's output.
The protocol relies heavily on the use of random bits that are shared between some, but not all, parties.
Such a resource can clearly not exist in the two-party setting.  
While it can be shown that perfect information security is unattainable by two-party protocols \cite{chor1989zero,bar1993privacy}, quantitatively
it is not clear just how much information the parties must ``leak" to each other to compute $f$. 
The quantitative answer depends on the model in which the leakage occurs, and whether quantum computation is allowed \cite{Klauck04}.
Information complexity answers this question in the strongest possible sense for classical protocols: the parties are allowed to use
private randomness to help them ``hide" their information, and the information revealed is measured on average. Thus an information 
complexity lower bound of $I$ on a problem implies that the {\em average} (as opposed to worst-case) amount of information revealed to 
the parties is at least $I$. 
 
As mentioned above, the information complexity is always upper bounded by the communication complexity of $f$. The converse is 
not known to be true. Moreover, lower bound techniques for communication complexity do not readily translate into lower bound techniques
for information complexity. The key difference is that a low-information protocol is not limited in the amount of communication
it uses, and thus rectangle-based communication bounds do not readily convert into information lower bounds. No general technique
has been known to yield sharp information complexity lower bounds. A linear lower bound on the communication complexity 
of the disjointness function has been shown in \cite{Raz92}. An  information-theoretic proof of this lower bound \cite{BaryossefJKS04} can be 
adapted  to prove a linear {\em information} lower bound on disjointness \cite{BravermanInteractive11}. 
One general technique for obtaining (weak) information complexity lower bounds was introduced in \cite{BravermanInteractive11}, where it has been 
shown that any function that has $I$ bits of information complexity, has communication complexity bounded by $2^{O(I)}$. 
This immediately implies that the information complexity of a function $f$ is at least the log of its communication 
complexity ($IC(f)\geq \Om(\log(CC(f)))$). In fact, this result easily follows from the stronger result we prove below 
(Theorem \ref{thm_pi'}).

\subsection{Our results} \label{our_res}

In this paper we prove that the discrepancy method -- a general communication complexity lower bound technique -- 
generalizes to information complexity. The discrepancy of $f$ with respect to a distribution $\mu$ on inputs, denoted 
$Disc_\mu(f)$, measures how ``unbalanced" $f$ can get on any rectangle, where the balancedness is measured with 
respect to $\mu$:
\begin{equation}
\label{eq:1}
\Disc \eqdef \max_{\text{rectangles }R}
 \bigg{|} \Pr_\mu[f(x,y) = 0 \wedge (x,y)\in R] - \Pr_\mu[f(x,y) = 1 \wedge (x,y)\in R] \bigg{|}. \nonumber 
\end{equation}
A well-known lower bound (see e.g \cite{Kushilevitz1997}) asserts that the distributional communication complexity of $f$,
denoted $D^\mu_{1/2-\eps}(f)$,
when required to predict $f$ with advantage $\ve$ over a random guess (with respect to $\mu$), 
is bounded from below by $\Om(\log 1/\Disc)$:
$$ D^\mu_{1/2-\eps}(f) \geq \log(2\eps/\Disc).$$
Note that the lower bound holds even if we are merely trying to get an advantage of $\ve=\sqrt{\Disc}$ over random guessing in 
computing $f$. We prove that the information complexity of computing $f$ with probability $9/10$ with respect to $\mu$ is also 
bounded from below by $\Om(\log(1/\Disc))$. 

\begin{theorem}\label{thm_disc}
Let $f: \cX\times \cY \rightarrow \{0,1\}$ be a Boolean function and let $\mu$ be any probability distribution on 
$\cX\times \cY$. Then 
$$ \ICmu{f}{\mu}{1/10} \geq \Omega(\log(1/\Disc)).$$
\end{theorem}
\smallskip

\begin{remark}
The choice of $9/10$ is somewhat arbitrary. For randomized worst-case protocols, we may replace the success 
probability with $1/2 + \delta$ for a constant $\delta$, since repeating the protocol constantly many times ($1/\delta^2$) would yield the 
aforementioned success rate, while the information cost of the repeated protocol differs only by a constant factor from the original 
one. In particular, using prior-free information cost \cite{BravermanInteractive11} this implies $ \IC{f}{1/2-\de} \geq \Omega(\delta^2 \log(1/\Disc)).$
\end{remark}

In particular, Theorem~\ref{thm_disc} implies a linear lower bound on the information complexity of the inner product function 
$IP(x,y)=\sum_{i=1}^n x_i y_i \mod 2$, and on a random boolean function $f_r:\{0,1\}^n \times \{0,1\}^n \ra \{0,1\}$, expanding 
the (limited) list of functions for which nontrivial information-complexity lower bounds are known:

\begin{corollary}
The information complexity $\ICmu{IP}{uniform}{1/10}$ of $IP(x,y)$ is $\Om(n)$. 
The information complexity $\ICmu{f_r}{uniform}{1/10}$ of a random function $f_r$ is $\Om(n)$, except with probability $2^{-\Omega(n)}$. 
\end{corollary}

We study the communication and information complexity of the Greater-Than function ($GT_n$) on numbers of length $n$. 
This is a very well-studied problem \cite{Smirnov88,miltersen1995data,Kushilevitz1997}. Only very recently the tight 
lower bound of $\Omega({\log n})$ in the public-coins probabilistic model was given by Viola \cite{Viola11}.
We show that the discrepancy of the $GT_n$ function is $\Omega(1/\sqrt{n})$:

\begin{lemma}
\label{lem:GT}
There exist a distribution $\mu_n$ on $\mathcal{X}\times\mathcal{Y}$ such that the discrepancy of $GT_n$ with respect to 
$\mu_n$ satisfies $$Disc_{\mu_n}(GT_n)<\frac{20}{\sqrt{n}}.$$ 
\end{lemma}

We defer the proof to the appendix.   Lemma \ref{lem:GT} provides an alternative (arguably simpler) proof  of Viola's  \cite{Viola11} lower bound on 
the {\em communication complexity} of $GT_n$. By Theorem \ref{thm_disc}, Lemma~\ref{lem:GT} immediately implies a lower bound on the {\em information complexity} of $GT_n$:

\begin{corollary}
 $\ICmu{GT_n}{\mu_n}{1/10} = \Omega(\log{n})$
 \end{corollary}
 
This settles the information complexity of the GT function, since this problem can be solved by a randomized protocol with	
$O(\log{n})$ communication (see \cite{Kushilevitz1997}). This lower bound is particularly interesting since it demonstrates
the first tight information-complexity lower bound for a natural function that is not linear.   

The key technical idea in the proof of Theorem \ref{thm_disc} is a new 
simulation procedure that allows us to convert any protocol that has information cost $I$ into a (two-round) protocol that has 
communication complexity $O(I)$ and succeeds with probability $>1/2 + 2^{-O(I)}$, yielding a $2^{-O(I)}$ advantage over 
random guessing. Combined with the discrepancy lower bound for communication complexity, this proves 
Theorem~\ref{thm_disc}.

\subsection{Comparison and connections to prior results}


The most relevant prior work is an article by Lee, Shraibman, and \v{S}palek \cite{Lee08adirect}. 
Improving on an earlier work of Shaltiel \cite{shaltiel2003towards}, Lee et al. show a direct product theorem for 
discrepancy, proving that the discrepancy of $f^{\otimes k}$  --- the  $k$-wise XOR of a function $f$ with itself  --- behaves 
as $Disc(f)^{\Omega(k)}$. This implies in particular that the communication complexity of $f^{\otimes k}$ scales
at least as $\Om(k\cdot \log Disc(f))$. Using the fact that the limit as $k\ra\infty$ of the amortized communication complexity of $f$ is equal to the information 
cost of $f$ \cite{BravermanR10}, the result of Lee et al. ``almost" implies the bound of Theorem~\ref{thm_disc}. Unfortunately, 
the amortized communication complexity in the sense of \cite{BravermanR10} is the amortized cost of $k$ copies of $f$, where {\em each} copy is 
allowed to err with some probability (say $1/10$). Generally speaking, this task is much easier than computing the XOR (which requires {\em all} copies to 
be evaluated correctly with high probability).
  Thus the lower bound that follows from Lee et al. applies to a more difficult problem, and 
does not imply the information complexity lower bound. 

Another generic approach one may try to take is to use compression results such as \cite{BarakBCR10} to lower bound the information cost from communication complexity lower bounds. The logic of such a proof would go as follows: ``Suppose there was a information-complexity-$I$ protocol $\pi$ for $f$, then if one can compress it into a low-communication 
protocol one may get a contradiction to the communication complexity lower bound $f$".  Unfortunately, all known compression results compress $\pi$ into a protocol $\pi'$ 
whose communication complexity depends on $I$ but also on $CC(\pi)$. Even for external information complexity (which is always greater than the internal information complexity, the bound obtained in \cite{BarakBCR10} 
is of the form $I_{ext}(\pi)\cdot {\text polylog} (CC(\pi))$. Thus compression results of this type cannot rule out protocols that have low information complexity 
but a very high (e.g. exponential) communication complexity.

Our result can be viewed as a weak compression result for protocols, where a protocol for computing $f$ that conveys $I$ bits of information is converted into 
a protocol that uses $O(I)$ bits of {\em communication} and giving an advantage of $2^{-O(I)}$ in computing $f$. This strengthens the result in 
\cite{BravermanInteractive11} where a compression to $2^{O(I)}$ bits of communication has been shown. We still do not know whether compression
to a protocol that uses $O(I)$ bits of communication and succeeds with high probability (as opposed to getting a small advantage over random) is possible. 

In a very recent breakthrough work of Kerenidis, Laplante, Lerays, Roland, and Xiao  \cite{kerenidis2012lower}, our main protocol 
played an important role in the proof that almost all known 
lower bound techniques for communication complexity (and not just discrepancy) apply to information complexity. The results 
of \cite{kerenidis2012lower} shed more light on the information complexity of many communication problems, and the question of whether interactive communication can be compressed. 



\end{section}


\begin{section}{Preliminaries}
\label{sec:prelim}

In an effort to make this paper as self-contained as possible, we provide some background on information theory 
and communication complexity, which is essential to proving our results. For further details and a more thorough
treatment of these subjects see \cite{BravermanR10} and references therein.

\paragraph{Notation.} We reserve capital letters for random variables and distributions, calligraphic letters for sets, 
and small letters for elements of sets. Throughout this paper, we often use the notation $|b$ to denote conditioning 
on the event $B=b$. Thus $A|b$ is shorthand for $A|B=b$. \\ 

\noindent We use the standard notion of \emph{statistical}/\emph{total variation} distance between two distributions.

\begin{definition}
Let $D$ and $F$ be two random variables taking values in a set $\mathcal{S}$. Their \emph{statistical distance} is
$| D-F | \eqdef \max_{\mathcal{T} \subseteq \mathcal{S}}(|\Pr[D \in \mathcal{T}] - \Pr[F \in \mathcal{T}]|) = \frac{1}{2} 
\sum_{s \in \mathcal{S}}|\Pr[D=s]-\Pr[F=s]|$
\end{definition}

\subsection{Information Theory}
\begin{definition}[Entropy]
 The \emph{entropy} of a random variable $X$ is $H(X) \eqdef \sum_x \Pr[X=x] \log(1/\Pr[X=x]).$
The \emph{conditional entropy} $H(X|Y)$ is defined as $\Ex{y \getsr Y}{ H(X|Y=y)}$.
\end{definition}


\begin{definition}[Mutual Information]
The \emph{mutual information} between two random variables $A,B$, denoted $I(A;B)$ is defined to be the quantity 
$H(A) - H(A|B) = H(B) - H(B|A).$ The \emph{conditional mutual information} $I(A;B |C)$ is $H(A|C) - H(A|BC)$.
\end{definition}


%
%

We also use the notion of \emph{divergence} (also known as Kullback-Leibler distance or relative entropy), which is a different way 
to measure the distance between two
distributions:
\begin{definition}[Divergence]
The informational divergence between two distributions is $$\Div{A}{B} \eqdef \sum_x A(x) \log(A(x)/B(x)).$$
\end{definition}


\begin{proposition}
Let $A,B,C$ be random variables in the same probability space. For every $a$ in the support of $A$ and $c$ in the 
support of $C$, let $B_a$ denote $B|A=a$ and $B_{ac}$ denote $B|A=a,C=c$. Then
$I(A;B | C) = \E_{a,c\in_R A,C} [\Div{B_{ac}}{B_c}]$.
\end{proposition}


\subsection{Communication Complexity} \label{section:communication}

We use the standard definitions of the computational model defined in~\cite{Yao79}. For complete details see section A of the appendix.

Given a communication protocol $\pi$, $\pi(x,y)$ denotes the concatenation of the public randomness with all the messages that are
sent during the execution of $\pi$. We call this the \emph{transcript} of the protocol. When referring to the random variable 
denoting the transcript, rather than a specific transcript, we will use the notation $\Pi(x,y)$ --- or simply $\Pi$ when $x$ and $y$ 
are clear from the context, thus $\pi(x,y)\getsr \Pi(x,y)$. 
When $x$ and $y$ are 
random variables themselves, we will denote the transcript by $\Pi(X,Y)$, or just $\Pi$. 

\begin{definition}[Communication Complexity notation]
For a function $f: \mathcal{X} \times \mathcal{Y} \rightarrow \Z_K$, a distribution $\mu$ supported on $\mathcal{X}
\times \mathcal{Y}$, and a parameter $\eps > 0$,  $D^\mu_{\eps}(f)$ denotes the communication complexity of the cheapest
deterministic protocol computing $f$ on inputs sampled according to $\mu$ with error $\eps$. 
\end{definition}

\begin{definition}[Combinatorial Rectangle]
A \it Rectangle \rm in $\mathcal{X} \times \mathcal{Y}$ is a subset $R\subseteq \mathcal{X} \times \mathcal{Y}$ which satisfies
$$ (x_1,y_1) \in R \text{ and } (x_2,y_2) \in R \Longrightarrow  (x_1,y_2) \in R$$

\end{definition}

\ignore{
\begin{definition}[Discrepancy of Boolean functions]
Let $f: \mathcal{X} \times \mathcal{Y} \rightarrow \{0,1\}$ be a Boolean function, $R$ be any rectangle, and $\mu$ be a probability 
distribution on $\mathcal{X} \times \mathcal{Y}$. Denote

$$ \DiscR = \bigg{|} \Pr_\mu[f(x,y) = 0 \wedge (x,y)\in R] - \Pr_\mu[f(x,y) = 1 \wedge (x,y)\in R] \bigg{|} $$

The \it Discrepancy \rm of $f$ with respect to $\mu$ is

$$ \Disc = \max_{R} \DiscR $$

\noindent where the maximum is taken over all rectangles $R$.
\end{definition}
}


\subsection{Information + Communication: The information cost of a protocol}

The following quantity, which is implicit in \cite{BaryossefJKS04} and was explicitly defined in \cite{BarakBCR10},
is the central notion of this paper. 

\begin{definition}\label{def:infoProt}
The \it (internal) information cost \rm of a protocol $\pi$ over inputs drawn from a distribution $\mu$ on $\cX\times\cY$, is given by:
$$\ICprot{\pi}{\mu}:=I(\Pi;X|Y)+I(\Pi;Y|X).$$
\end{definition}

\noindent
Intuitively, Definition~\ref{def:infoProt} captures how much the two parties learn about each other's inputs from the execution 
transcript of the protocol $\pi$. The first term captures
what the second player learns about $X$ from $\Pi$ -- the mutual information between the input $X$ and the transcript 
$\Pi$ given the input $Y$. Similarly, the second term captures
what the first player learns about $Y$ from $\Pi$.

Note that the information of a protocol $\pi$ depends on the prior distribution $\mu$, as the mutual information between 
the transcript $\Pi$ and the inputs depends on the prior distribution 
on the inputs. To give an extreme example, if $\mu$ is a singleton distribution, i.e. one with $\mu(\{(x,y)\}) = 1$ for some 
$(x,y)\in \cX\times \cY$, then $\ICprot{\pi}{\mu}=0$ for all possible 
$\pi$, as no protocol can reveal anything to the players about the inputs that they do not already know {\em a-priori}. 
Similarly, $\ICprot{\pi}{\mu}=0$ if $\cX=\cY$ and $\mu$ is supported on 
the diagonal $\{(x,x):x\in\cX\}$.
\smallskip
As expected, one can show that the communication cost $\CProt{\pi}$ of $\pi$ is an upper bound on its information cost 
over {\em any} distribution
$\mu$:

\begin{lemma} \label{lem:ICCC}
\cite{BravermanR10} For any distribution $\mu$, $\ICprot{\pi}{\mu}\le \CProt{\pi}$. 
\end{lemma}

\noindent On the other hand, as noted in the introduction, the converse need not hold.  
In fact, by \cite{BravermanR10}, getting a strict inequality in Lemma~\ref{lem:ICCC} is equivalent to violating the direct sum conjecture 
for randomized communication complexity.

\indent As one might expect, the information cost of a function $f$ with respect to $\mu$ and error $\rho$ is the least amount of 
information that needs to be revealed by a protocol computing $f$ with error $\le \rho$:
$$
\ICmu{f}{\mu}{\rho}:= \inf_{\pi:~ \P_\mu[\pi(x,y)\neq f(x,y)]\le \rho} \ICprot{\pi}{\mu}. 
$$
The (prior-free) information cost was defined in \cite{BravermanInteractive11} as the minimum amount of information that 
a worst-case error-$\rho$ randomized protocol can reveal against {\em all} possible prior distributions. 
$$
\IC{f}{\rho} := \inf_{\text{$\pi$ is a protocol with $\P[\pi(x,y)\neq f(x,y)]\le\rho$ for all $(x,y)$}} \max_\mu~~ \ICprot{\pi}{\mu}.
$$
This information cost matches the amortized {\em randomized} communication complexity of $f$ \cite{BravermanInteractive11}. It is 
clear that lower bounds on $\ICmu{f}{\mu}{\rho}$ {\em for any distribution $\mu$} also apply to $\IC{f}{\rho} $.

\end{section}


\begin{section}{Proof of Theorem \ref{thm_disc}}

To establish the correctness of Theorem \ref{thm_disc}, we prove the following theorem, which is the central result 
of this paper:

\begin{theorem}\label{thm_pi'}

Suppose that $\ICmu{f}{\mu}{1/10} = I_\mu$. Then there exist a protocol $\pi'$ such that 

\begin{itemize}

\item $\CProt{\pi'} = O(I_\mu)$.
\item $\P_{(x,y)\sim\mu}[\pi'(x,y) = f(x,y)] \geq 1/2 + 2^{-O(I_\mu)}$

\end{itemize}

\end{theorem}

\hspace{1in}

\noindent We first show how Theorem \ref{thm_disc} follows from Theorem \ref{thm_pi'}: \\

\noindent \bf{Proof of Theorem \ref{thm_disc}.}\rm \hspace{0.05in}  Let $f, \mu$ be as in theorem 
\ref{thm_disc}, and let $\ICmu{f}{\mu}{1/10} = I_\mu$. By theorem \ref{thm_pi'},
there exists a protocol $\pi'$ computing $f$ with error probability $1/2 - 2^{-O(I_\mu)}$ using $O(I_\mu)$ 
bits of communication. Applying the discrepancy lower bound for communication complexity we obtain

\begin{equation}
O(I_\mu) \geq D^\mu_{1/2-2^{-O(I_\mu)}}(f) \geq \log(2\cdot 2^{-O(I_\mu)}/\Disc)
\end{equation}

\noindent which after rearranging gives $I_\mu \geq \Omega(\log(1/\Disc))$, as desired.

\hspace{.5in}

We now turn to prove Theorem \ref{thm_pi'}. The main step is the following sampling lemma.

\begin{lemma}\label{lem_sample}

Let $\mu$ be any distribution over a universe $\cU$ and let $I\ge 0$ be a parameter that is known 
to both $A$ and $B$. Further, let $\nu_A$ and $\nu_B$ be two 
distributions over $\cU$ such that $\Div{\mu}{\nu_A}\le I$ and $\Div{\mu}{\nu_B}\le I$. 
The players are each given a pair of real functions $(p_A,q_A)$, $(p_B,q_B)$,  $p_A,q_A,p_B,q_B:\cU\ra [0,1]$ such 
that for all $x\in \cU$, $\mu(x) = p_A(x)\cdot p_B(x)$, $\nu_A(x)=p_A(x)\cdot q_A(x)$, and $\nu_B(x)=p_B(x)\cdot q_B(x)$.
Then there is a (two round) sampling protocol $\Pi_1=\Pi_1(p_A, p_B, q_A, q_B,I)$ 
which has the following properties:

\begin{enumerate}
\item 
at the end of the protocol, the players either declare that the protocol ``fails", or output $x_A\in\cU$ and 
$x_B\in \cU$, respectively (``success").
\item 
let $\suc$ be the event that the players output ``success". Then $\suc \Rightarrow x_A = x_B$, and \newline 
$0.9 \cdot 2^{-50{(I+1)}} \leq \Pr[\suc] \leq 2^{-50(I+1)}$.
\item 
if $\mu_1$ is the distribution of $x_A$ conditioned on $\suc$, then
$|\mu-\mu_1|<2/9$.
\end{enumerate} 

\noindent Furthermore, $\Pi_1$ can be ``compressed" to a protocol $\Pi_2$ such that $\CProt{\Pi_2} = 211 I + 1$,
whereas
$|\Pi_1 - \Pi_2| \leq 2^{-59 I}$ (by an abuse of notation, here we identify $\Pi_i$ with the random variable representing 
its output).

\end{lemma}

\hspace{.5in}

We will use the following technical fact about the information divergence of distributions.

\begin{claim}[3]{[Claim 5.1 in \cite{BravermanInteractive11}]} \label{lem_kl}
\label{cl:div1}
Suppose that $\Div{\mu}{\nu}\le I$. Let $\ve$ be any parameter. Then 
$$
\mu\left\{x:~2^{(I+1)/\ve}\cdot \nu(x) <\mu(x)\right\} < \ve.
$$
\end{claim}

\noindent For completeness, we repeat the proof in the appendix.. 

\begin{proof}[\bf Proof of Lemma \ref{lem_sample} \rm]
Throughout the execution of $\Pi_1$, Alice and Bob interpret their shared random tape as a source of 
points $(x_i,\al_i,\be_i)$ uniformly distributed in $\cU\times [0,2^{50(I+1)}] \times [0,2^{50(I+1)}]$. 
Alice and Bob consider the first $T= |\cU|\cdot 2^{100(I+1)}\cdot60I$ such points. 
Their goal will be to discover the first index $\tau$ such that $\al_\tau\le p_A(x_\tau)$ and $\be_\tau\le p_B(x_\tau)$ 
(where they wish to find it using a minimal amount of communication, even if they are most likely to fail). 
First, we note that the probability that an index $t$ satisfies $\al_t\le p_A(x_t)$ and $\be_t \le p_B(x_t)$ 
is exactly $1/|\cU|2^{50(I+1)}2^{50(I+1)} = 1/|\cU|2^{100(I+1)}$. Hence the probability that $\tau>T$ 
(i.e. that $x_\tau$ is not among the $T$ points considered) is bounded by
\smallskip 
\begin{equation} \label{a_nonempty}
\left(1-1/|\cU|2^{100(I+1)}\right)^{T} < e^{-T/|\cU|2^{100(I+1)}} = e^{-60I} < 2^{-60I}
\end{equation}

Denote by $\cA$ the following set of indices $\cA := \{i\le T:~\al_i\le p_A(x_i) \text{ and }\be_i\le 2^{50(I+1)}\cdot q_A(x_i)\}$, 
the set of potential candidates for $\tau$ from A's viewpoint. Similarly, denote 
$\cB := \{i\le T:~\al_i\le 2^{50(I+1)}\cdot q_B(x_i) \text{ and }\be_i\le  p_B(x_i)\}$.

The protocol $\Pi_1$ is very simple. 
Alice takes her bet on the first element $a \in \cA$ and sends it to Bob. 
Bob outputs $a$ only if (it just so happens that) $\be_\tau\le p_B(a)$. The details are given in Figure~\ref{figure:pi1} in the appendix. \\

We turn to analyze $\Pi_1$. Denote the set of ``Good" elements by $$\G\eqdef \{ x:~2^{50(I+1)}\cdot \nu_A(x) 
\geq \mu(x) \; \text{ and } \; 2^{50(I+1)}\cdot \nu_B(x) \geq \mu(x)\}\}.$$ Then by Claim~\ref{cl:div1}, $\mu(\G) \geq 
48/50 = 24/25$. The following claim asserts that if it succeeds, the output of $\Pi_1$ has the ``correct" distribution
on elements in $\G$. 

\begin{proposition}\label{c1}
Assume $\cA$ is nonempty. Then for any $x_i \in \cU$, the probability that $\Pi_1$ outputs $x_i$ is at most 
$\mu(x_i)\cdot 2^{-50(I+1)}$. If $x_i \in \G$, then this probability is \em{ exactly} $\mu(x_i)\cdot 2^{-50(I+1)}$.
\end{proposition}

\begin{proof}
Note that if $\cA$ is nonempty, then for any $x_i\in \cU$, the probability that $x_i$ is the first element in 
$\cA$ (i.e, $a = x_i$) is $p_A(x_i)q_A(x_i) = \nu_A(x_i)$. By construction, the probability that $\beta_i\leq p_B(a)$ is 
$\min\{p_B(x_i)/(2^{50{(I+1)}}q_A(x_i)), 1\}$, and thus
$$ \Pr[\Pi_1 \text{ outputs } x_i] \leq
p_A(x_i)q_A(x_i)\cdot\frac{p_B(x_i)}{2^{50(I+1)}q_A(x_i)} = \mu(x_i)\cdot 2^{-50(I+1)}.$$
On the other hand, if $x_i \in \G$, then  we know that 
$
p_B(x_i)/q_A(x_i) = \mu(x_i)/\nu_A(x_i) \le 2^{50(I+1)},
$
and so the probability that $\beta_i\leq p_B(a)$ is \it exactly \rm $p_B(x_i)/(2^{50{(I+1)}}q_A(x_i))$. 
Since $\Pr[\Pi_1 \text{ outputs } x_i]  = \Pr[a = x_i]\Pr[\beta_i\leq p_B(a)]$ (assuming $\cA$ is nonempty),
we conclude that:
$$ x_i \in \G \Longrightarrow \Pr[\Pi_1 \text{ outputs } x_i] = 
p_A(x_i)q_A(x_i)\cdot\frac{p_B(x_i)}{2^{50(I+1)}q_A(x_i)} = \mu(x_i)\cdot 2^{-50(I+1)}.$$ 
\end{proof}

We are now ready to estimate the success probability of the protocol. 

\begin{proposition}\label{s_bound}
Let $\suc$ denote the event that $\cA\neq 0$ and  $a \in \cB$ (i.e, that the protocol succeeds). Then
$$0.9 \cdot 2^{-50{(I+1)}} \leq \Pr[\suc] \leq 2^{-50(I+1)}.$$
\end{proposition}

\begin{proof}

Using Proposition~\ref{c1}, we have that

\begin{eqnarray}
\Pr[\suc] &\leq& \P[a\in \cB \; |\; \cA \neq \emptyset]  = 
\sum_{i\in\cU}\Pr[a = x_i]\Pr[\beta_i\leq p_B(a)] \leq
\\ \nonumber &\leq&  \sum_{i\in\cU} \mu(x_i)\cdot 2^{-50(I+1)} = 2^{-50(I+1)} \label{eq4}
\end{eqnarray}

\noindent For the lower bound, we have

\begin{eqnarray}
\Pr[\suc] &\geq& \Pr[\beta_i\leq p_B(a)  \; |\; \cA \neq \emptyset]\cdot\Pr[\cA \neq \emptyset]  \nonumber \geq \\
&\geq& (1-2^{-60I})\bigg{(}\sum_{i\in\cU}\Pr[a = x_i]\Pr[\beta_i\leq p_B(a)]\bigg{)} \geq \nonumber \\
&\geq& (1-2^{-60I})\bigg{(}\sum_{i\in\G}\Pr[a = x_i]\Pr[\beta_i\leq p_B(a)]\bigg{)} = \nonumber \\
 &=& (1-2^{-60I})\bigg{(} 2^{-50(I+1)}\sum_{i\in\G} \mu(x_i) \bigg{)} = 
 (1-2^{-60I})\bigg{(} 2^{-50(I+1)}\mu(\G) \bigg{)} \geq \nonumber  \\
 &\geq&  \frac{24}{25}(1-2^{-60I})2^{-50{(I+1)}} \geq 0.9 \cdot 2^{-50{(I+1)}} \label{eq5}
\end{eqnarray}

\noindent where the equality follows again from claim \ref{c1}. This proves the second claim of Lemma~\ref{lem_sample}.
\end{proof}

The following claim asserts that if $\suc$ occurs, then the distribution of $a$ is indeed close to $\mu$. \\

\noindent \bf Claim 4. \rm
Let $\mu_1$ be the distribution of $a|\suc$. Then $|\mu_1 - \mu| \leq 2/9$.

\noindent \it Proof. \rm The claim follows directly from proposition \ref{s_bound}. We defer the proof to the appendix. \\

We turn to the ``Furthermore" part of of Lemma~\ref{lem_sample}. The protocol $\Pi_1$ satisfies the premises of the lemma, except  
it has a high communication cost. This is due to the fact that Alice explicitly sends $a$ to Bob. To reduce the communication,
Alice will instead send $O(I)$ random hash values of $a$, and Bob will add corresponding consistency constraints 
to his set of candidates. The final protocol $\Pi_2$ is given in Figure~\ref{figure:pi2}.


\begin{figure}[h!tb]
\begin{tabular}{|l|}
\hline
\begin{minipage}{\algwidth}
\vspace{1ex}
\begin{center}
\textbf{Information-cost sampling protocol $\Pi_2$}
\end{center}
\vspace{0.5ex}
\end{minipage}\\
\hline
\begin{minipage}{\algwidth}
\vspace{1ex}
\begin{enumerate}
    \item Alice computes the set $\cA$.  Bob  computes the set $\cB$. 
    \item If $\cA=\emptyset$, the protocol fails. Otherwise, Alice finds the first element $a \in \cA$ and sets $x_A = a$. She then computes $d=\lceil 211 I \rceil$ 
    random hash values $h_1(a),\ldots,h_d(a)$, where the hash functions are evaluated using public randomness. 
    \item Alice sends the values $\{h_j(a)\}_{1\le j\le d}$ to Bob. 
    \item \label{find2} Bob finds the first index $\tau$ such that there is a $b\in \cB$ for which $h_j(b)=h_j(a)$ for $j=1..d$ (if such 
    an $\tau$ exists). Bob outputs $x_B = x_\tau$. If there is no such index, the protocol fails.
    \item Bob outputs $x_B$ (``success"). 
    \item Alice outputs $x_A$. 
\end{enumerate}
\vspace{0.3ex}
\end{minipage}\\
\hline
\end{tabular}
\caption{The sampling protocol $\Pi_2$ from Lemma~\ref{lem_sample} }\label{figure:pi2}
\end{figure}


Let $\cE$ denote the event that in step \ref{find} of the protocol, Bob finds an element $x_i\neq a$ 
(that is, the probability that the protocol outputs ``success" but $x_A \neq x_B$). We upper bound the 
probability of $\cE$. 
Given $a\in \cA$ and $x_i\in \cB$ such that $a\neq x_i$, the probability (over possible choices of the hash 
functions) that $h_j(a)=h_j(x_i)$ for $j=1..d$ is $2^{-d} \leq 2^{-211 I}$. 
For any $t$, $\P[t\in \cB] \leq \frac{1}{|\cU|}\sum_{x_i\in \cU}p_B(x_i)q_B(x_i)\cdot2^{50(I+1)} =
\frac{1}{|\cU|}\sum_{x_i\in \cU}\nu_B(x_i)\cdot2^{50(I+1)} =
2^{50(I+1)}/|\cU|$. Thus, by a union bound we have

\begin{eqnarray}
\P[\cE] &\leq& \P[\exists x_i\in \cB \; s.t \; x_i\neq a \; \wedge \; h_j(a)=h_j(x_i) \; \forall \; j=1,\ldots,d] \leq \nonumber \\
&\leq& T\cdot 2^{50(I+1)}\cdot 2^{-d}/|\cU| = 2^{150(I+1)}\cdot60 I \cdot 2^{-211 I} \ll 2^{-60 I}.
\end{eqnarray}

By a slight abuse of notation, let $\Pi_2$ be the distribution of $\Pi_2$'s output. Similarly, denote by 
$\Pi_1$ the distribution of the output of protocol $\Pi_1$. Note that if $\cE$ does not occur, then the outcome of the execution of  $\Pi_2$ 
 is identical to the outcome of  $\Pi_1$. 
 Since $\P[\cE] \leq 2^{-60 I}$, we have
$$ |\Pi_2 - \Pi_1| = \Pr[\cE] \cdot |[\Pi_2|\cE]-[\Pi_1|\cE]| \le 2\cdot 2^{-60I} \ll 2^{-59 I}$$ 

\noindent which finishes the proof of the lemma.

\end{proof}

\begin{remark}
The communication cost of the sampling protocol $\Pi_2$ can be reduced from $O(I_\mu)$ to $O(1)$ (more precisely, to only two bits) in the following way:
Instead of having Alice compute the hash values privately and send them to Bob in step 2 and 3 of the protocol, the players can use their
shared randomness to sample $d=O(I_\mu)$ random hash values $h_1(b_1),\ldots,h_d(b_d)$ (where the $b_i$'s are random independent strings in $\cU$), and Alice will   only send Bob a single bit indicating whether those hash values match the hashing of her string $a$ (i.e, $h_i(b_i) = h_i(a)$ for all $i\in [d]$).
In step 4 Bob will act as before, albeit comparing the hashes of his candidate $b$ to the random hashes $h_i(b_i)$, and output success ("1") if the hashes match.
Note that this modification incurs an additional loss of $2^{-d}=2^{-O(I_\mu)}$ in the success probability of the protocol (as this is the probability that $h_i(b_i) = h_i(a)$ for all $i\in [d]$),
but since the success probability we are shooting for is already of the order $2^{-O(I_\mu)}$, we can afford this loss.
This modification was observed in \cite{kerenidis2012lower}.
\end{remark}

With Lemma \ref{lem_sample} in hand, we are now ready to prove Theorem \ref{thm_pi'}.


\noindent \bf Proof of Theorem \ref{thm_pi'}. \rm Let $\pi$ be a protocol that realizes the value $I_\mu:=\ICmu{f}{\mu}{1/10}$.
In other words, $\pi$ has an error rate of at most $1/10$ and information cost of at most $I_\mu$ with respect to $\mu$. 
Denote by $\pi_{xy}$ the random variable that represents that transcript $\pi$ given the inputs $(x,y)$, and by 
$\pi_x$ (resp. $\pi_y$) the protocol conditioned on only the input $x$ (resp. $y$). We denote by $\pi_{XY}$ the 
transcripts where $(X,Y)$ are also a pair of random variables. By Claim \ref{lem_kl}, we know that

\begin{eqnarray}\label{div}
I_\mu = I(X; \pi_{XY}|Y)+I(Y; \pi_{XY}|X) = \E_{(x,y)\sim \mu} [\Div{\pi_{xy}}{\pi_x}+\Div{\pi_{xy}}{\pi_y}].
\end{eqnarray}

Let us now run the sampling algorithm $\Pi_1$ from Lemma~\ref{lem_sample}, with the distribution $\mu$
taken to be $\pi_{xy}$, the distributions $\nu_A$ and $\nu_B$ taken to be $\pi_x$ and $\pi_y$ respectively, 
and $I$ taken to be $20 I_\mu$.

At each node $v$ of the protocol tree that is owned by player $X$ let $p_0(v)$ and $p_1(v)=1-p_0(v)$ denote
the probabilities that the next bit sent by $X$ is $0$ and $1$, respectively. For nodes owned by player $Y$, let $q_0(v)$ 
and $q_1(v)=1-q_0(v)$ denote the probabilities that the next bit sent by $Y$ is $0$ and $1$, respectively, {\em as 
estimated by player $X$ given the input $x$}. For each leaf $\ell$ let $p_X(\ell)$ be the product
of all the values of $p_b(v)$ from the nodes that are owned by $X$ along the path from the root to $\ell$; let $q_X(\ell)$ 
be the product of all the values of $q_b(v)$ from the nodes that are owned by $Y$ along the path from the root to $\ell$. 
The values $p_Y(\ell)$ and $q_Y(\ell)$ are 
defined similarly. For each $\ell$ we have $\P[\pi_{xy}=\ell] = p_X(\ell)\cdot p_Y(\ell)$, $\P[\pi_{x} = \ell] = 
p_X(\ell)\cdot q_X(\ell)$, and $\P[\pi_{y} = \ell] = p_Y(\ell)\cdot q_Y(\ell)$.
Thus we can apply Lemma~\ref{lem_sample} so as to obtain the following protocol $\pi'$ for computing $f$:

\begin{itemize}
\item If $\Pi_1$ fails, we return a random unbiased coin flip.
\item If $\Pi_1$ succeeds, we return the final bit of the transcript sample $T$. Denote this bit by $T_{out}$.
\end{itemize}

To prove the correctness of the protocol, let $\cZ$ denote the event that both $\Div{\pi_{xy}}{\pi_x} \le 20 I_\mu$ and 
$\Div{\pi_{xy}}{\pi_y} \le 20 I_\mu$. By (\ref{div}) and Markov inequality, $\Pr[\cZ] \geq 19/20$ (where the probability 
is taken with respect to $\mu$). Denote by $\delta$ the probability that $\Pi_1$ succeeds. By the assertions of Lemma~\ref{lem_sample}, 
$\delta \geq 0.9\cdot 2^{-50(I+1)}$.  Furthermore, if $\Pi_1$ succeeds, then we have $|T-\pi_{xy}|\leq2/9$, 
which in particular implies that $\P[T_{out} = \pi_{out}] \geq 7/9$. Finally, $\P[ \pi_{out} = f(x,y)] \geq 9/10$, since $\pi$ 
has error at most $1/10$ with respect to $\mu$. Now, let $\cW$ denote the indicator variable whose value is $1$ iff 
$\pi'(x,y) = f(x,y)$. Putting together the above,

\begin{eqnarray}
\E[\cW \; | \; \cZ] = (1-\delta)\cdot \frac{1}{2} + \delta\cdot \left(\frac{7}{9}-\frac{1}{10}\right) > \frac{1}{2} + 
\delta\cdot\frac{1}{6} > \frac{1}{2} + \frac{1}{8}\cdot 2^{-50(I+1)}.
\end{eqnarray}

\noindent On the other hand, note that by lemma \ref{lem_sample} the probability that $\Pi_1$ succeeds is at most 
$2^{-50(I+1)}$ (no matter how large $\Div{\pi_{xy}}{\pi_x}$ and $\Div{\pi_{xy}}{\pi_y}$ are!), and so
$
\E[\cW \; | \; \neg\cZ] \geq (1-2^{-50(I+1)})/2.
$

\noindent Hence we conclude that
\smallskip
\begin{eqnarray}
\E[\cW] &=& \E[\cW \; | \; \cZ]\cdot\P[\cZ] + \E[\cW \; | \; \neg\cZ]\cdot\P[\neg\cZ] \geq \left(\frac{1}{2} + 
\frac{1}{8}\cdot 2^{-50(I+1)}\right)\cdot \frac {19}{20} + \nonumber  \\  &+&
\left(1-2^{-50(I+1)}\right)\cdot\frac{1}{2}\cdot\frac{1}{20}
\geq \frac{1}{2} + \frac{1}{12}\cdot 2^{-50(I+1)} > \frac{1}{2} + \frac{1}{12}\cdot 2^{-1000(I_\mu+1)}. \nonumber 
\end{eqnarray}

Finally, Lemma~\ref{lem_sample} asserts that $|\Pi_1 - \Pi_2| < 2^{-59 I}$. Thus if we replace $\Pi_1$ by $\Pi_2$ 
in the execution of protocol $\pi'$,  
the success probability decreases by at most $2^{-59 I}\ll \frac{1}{12}\cdot 2^{-50(I+1)}$. 
Furthermore, the amount of communication used by $\pi'$ is now $$ 211 I = 4220 I_\mu = O(I_\mu).$$

\noindent Hence we conclude that with this modification, $\pi'$ has the following properties:
\begin{itemize}
\item $\CProt{\pi'} = 4220\cdot I_\mu$;
\item $\P_{(x,y)\sim\mu}[\pi'(x,y) = f(x,y)] \geq 1/2 + 2^{-1000(I_\mu+1)-4}$;
\end{itemize}

\noindent which completes the proof.

\begin{remark}
Using similar techniques, it was recently shown in \cite{BravermanInteractive11} that any function $f$ 
whose information complexity is $I$ has communication cost at most $2^{O(I)}$ \footnote{More precisely, it shows 
that for any distribution $\mu$, $D_{\ve + \delta}^\mu(f)  = 2^{O(1+\ICmu{f}{\mu}{\ve}/\delta^2)}$.},
thus implying that $IC(f)\geq \Om(\log(CC(f)))$. We note that this result can be easily derived (up to constant factors) 
from Theorem \ref{thm_pi'}. Indeed, applying the ``compressed" protocol $2^{O(I)}\log(1/\eps)$ independent 
times and taking a majority vote guarantees an error of at most $\eps$ (by a standard Chernoff bound\footnote{See N.Alon 
and J. Spencer, "The Probabilistic Method" (Third Edition) ,Corollary A.1.14, p.312.}), with communication 
$O(I)\cdot2^{O(I)} = 2^{O(I)}$. Thus, our result is strictly stronger than the former one.
\end{remark}

\end{section}


\subsection*{Acknowledgments}
We thank Ankit Garg and several anonymous reviewers from RANDOM 12'  for their useful comments and helpful discussions.

\bibliographystyle{alpha}
\bibliography{refs}

\newcommand{\etalchar}[1]{$^{#1}$}
\begin{thebibliography}{BYCKO93}

\bibitem[AL11]{asharov2011full}
G.~Asharov and Y.~Lindell.
\newblock A full proof of the bgw protocol for perfectly-secure multiparty
  computation.
\newblock {\em Advances in Cryptology—CRYPTO 2011}, 2011.

\bibitem[BBCR10]{BarakBCR10}
Boaz Barak, Mark Braverman, Xi~Chen, and Anup Rao.
\newblock How to compress interactive communication.
\newblock In {\em Proceedings of the 42nd Annual ACM Symposium on Theory of
  Computing}, 2010.

\bibitem[BOGW88]{ben1988completeness}
M.~Ben-Or, S.~Goldwasser, and A.~Wigderson.
\newblock Completeness theorems for non-cryptographic fault-tolerant
  distributed computation.
\newblock In {\em Proceedings of the twentieth annual ACM symposium on Theory
  of computing}, pages 1--10. ACM, 1988.

\bibitem[BR10]{BravermanR10}
Mark Braverman and Anup Rao.
\newblock Information equals amortized communication.
\newblock {\em CoRR}, abs/1106.3595, 2010.

\bibitem[BR11]{braverman2011information}
M.~Braverman and A.~Rao.
\newblock Information equals amortized communication.
\newblock {\em Arxiv preprint arXiv:1106.3595}, 2011.

\bibitem[Bra11]{BravermanInteractive11}
Mark Braverman.
\newblock Interactive information complexity.
\newblock {\em Electronic Colloquium on Computational Complexity (ECCC)},
  18:123, 2011.

\bibitem[BYCKO93]{bar1993privacy}
R.~Bar-Yehuda, B.~Chor, E.~Kushilevitz, and A.~Orlitsky.
\newblock Privacy, additional information and communication.
\newblock {\em Information Theory, IEEE Transactions on}, 39(6):1930--1943,
  1993.

\bibitem[BYJKS04]{BaryossefJKS04}
Ziv Bar-Yossef, T.~S. Jayram, Ravi Kumar, and D.~Sivakumar.
\newblock An information statistics approach to data stream and communication
  complexity.
\newblock {\em Journal of Computer and System Sciences}, 68(4):702--732, 2004.

\bibitem[CK89]{chor1989zero}
B.~Chor and E.~Kushilevitz.
\newblock A zero-one law for boolean privacy.
\newblock In {\em Proceedings of the twenty-first annual ACM symposium on
  Theory of computing}, pages 62--72. ACM, 1989.

\bibitem[CSWY01]{ChakrabartiSWY01}
Amit Chakrabarti, Yaoyun Shi, Anthony Wirth, and Andrew Yao.
\newblock Informational complexity and the direct sum problem for simultaneous
  message complexity.
\newblock In Bob Werner, editor, {\em Proceedings of the 42nd Annual IEEE
  Symposium on Foundations of Computer Science}, pages 270--278, Los Alamitos,
  CA, October ~14--17 2001. IEEE Computer Society.

\bibitem[HJMR07]{harsha2007communication}
P.~Harsha, R.~Jain, D.~McAllester, and J.~Radhakrishnan.
\newblock The communication complexity of correlation.
\newblock In {\em Computational Complexity, 2007. CCC'07. Twenty-Second Annual
  IEEE Conference on}, pages 10--23. IEEE, 2007.

\bibitem[Jai10]{jain2010strong}
R.~Jain.
\newblock A strong direct product theorem for two-way public coin communication
  complexity.
\newblock {\em Arxiv preprint arXiv:1010.0846}, 2010.

\bibitem[JSR08]{jain2008optimal}
R.~Jain, P.~Sen, and J.~Radhakrishnan.
\newblock Optimal direct sum and privacy trade-off results for quantum and
  classical communication complexity.
\newblock {\em Arxiv preprint arXiv:0807.1267}, 2008.

\bibitem[Kla04]{Klauck04}
Hartmut Klauck.
\newblock Quantum and approximate privacy.
\newblock {\em Theory Comput. Syst.}, 37(1):221--246, 2004.

\bibitem[Kla10]{klauck2010strong}
H.~Klauck.
\newblock A strong direct product theorem for disjointness.
\newblock In {\em Proceedings of the 42nd ACM symposium on Theory of
  computing}, pages 77--86. ACM, 2010.

\bibitem[KLL{\etalchar{+}}12]{kerenidis2012lower}
I.~Kerenidis, S.~Laplante, V.~Lerays, J.~Roland, and D.~Xiao.
\newblock Lower bounds on information complexity via zero-communication
  protocols and applications.
\newblock {\em Arxiv preprint arXiv:1204.1505}, 2012.

\bibitem[KN97]{Kushilevitz1997}
Eyal Kushilevitz and Noam Nisan.
\newblock {\em Communication complexity}.
\newblock Cambridge University Press, New York, 1997.
\newblock 96012840 96012840 Eyal Kushilevitz, Noam Nisan.

\bibitem[KSDW04]{klauck2004quantum}
H.~Klauck, R.~Spalek, and R.~De~Wolf.
\newblock Quantum and classical strong direct product theorems and optimal
  time-space tradeoffs.
\newblock In {\em Foundations of Computer Science, 2004. Proceedings. 45th
  Annual IEEE Symposium on}, pages 12--21. IEEE, 2004.

\bibitem[LSS08]{Lee08adirect}
T.~Lee, A.~Shraibman, and R.~Spalek.
\newblock A direct product theorem for discrepancy.
\newblock In {\em Computational Complexity, 2008. CCC'08. 23rd Annual IEEE
  Conference on}, pages 71--80. IEEE, 2008.

\bibitem[MNSW95]{miltersen1995data}
P.B. Miltersen, N.~Nisan, S.~Safra, and A.~Wigderson.
\newblock On data structures and asymmetric communication complexity.
\newblock In {\em Proceedings of the twenty-seventh annual ACM symposium on
  Theory of computing}, pages 103--111. ACM, 1995.

\bibitem[Raz92]{Raz92}
Alexander~A. Razborov.
\newblock On the distributional complexity of disjointness.
\newblock {\em Theor. Comput. Sci.}, 106(2):385--390, 1992.

\bibitem[Sha03]{shaltiel2003towards}
R.~Shaltiel.
\newblock Towards proving strong direct product theorems.
\newblock {\em Computational Complexity}, 12(1):1--22, 2003.

\bibitem[Smi88]{Smirnov88}
D~V Smirnov.
\newblock Shannon's information methods for lower bounds for probabilistic
  communication complexity.
\newblock Master's thesis, Moscow State University, 1988.

\bibitem[Vio11]{Viola11}
Emanuele Viola.
\newblock The communication complexity of addition.
\newblock {\em Electronic Colloquium on Computational Complexity (ECCC)},
  18:152, 2011.

\bibitem[Yao79]{Yao79}
Andrew Chi-Chih Yao.
\newblock Some complexity questions related to distributive computing
  (preliminary report).
\newblock In {\em STOC}, pages 209--213, 1979.

\end{thebibliography}

\appendix

\section{Communication Complexity}

Let $\mathcal{X}, \mathcal{Y}$ denote the set of possible inputs to the two players, who we name A and B.  We
 view a \emph{private coins protocol} for computing a function $f: \mathcal{X} \times \mathcal{Y}
\rightarrow \Z_K$ as a rooted tree with the following structure:
\begin{itemize}

\item Each non-leaf node is \emph{owned} by A or by B.
\item Each non-leaf node owned by a particular player has a set of children that are owned by the other player. 
Each of these children is labeled by a binary string, in such a way that this coding is prefix free: no child has a 
label that is a prefix of another child. 

\item Every node is associated with a function mapping $\mathcal{X}$ to distributions on children of the node 
and a function mapping $\mathcal{Y}$ to distributions on children of the node.

\item The leaves of the protocol are labeled by output values.

\end{itemize}


\ignore{
\begin{figure}[h!tb]
\begin{tabular}{|l|}
\hline
\begin{minipage}{\algwidth}
\vspace{1ex}
\begin{center}
\textbf{Generic Communication Protocol}
\end{center}
\vspace{0.5ex}
\end{minipage}\\
\hline
\begin{minipage}{\algwidth}
\vspace{1ex}
\begin{enumerate}
    \item Set $v$ to be the root of the protocol tree.
    \item If $v$ is a leaf, the protocol ends and outputs the value in the label of $v$. Otherwise, the player owning 
    $v$ samples a child of $v$ according to the distribution associated with her input for $v$ and sends the label 
    to indicate which child was sampled.
    \item Set $v$ to be the newly sampled node and return to the previous step.
\end{enumerate}
\vspace{0.3ex}
\end{minipage}\\
\hline
\end{tabular}
\caption{A communication protocol.}\label{figure:pi1}
\end{figure}
}

A public coin protocol is a distribution on private coins protocols, run by first using shared randomness to sample an
index $r$ and then running the corresponding private coin protocol $\pi_r$. Every private coin protocol is thus a
public coin protocol. The protocol is called deterministic if all distributions labeling the nodes have support size $1$.

\begin{definition}
The \emph{communication cost} (or communication complexity) of a public coin protocol $\pi$, denoted $\CProt{\pi}$, 
is the maximum number of bits that can be transmitted in any run of the protocol.
\end{definition}

\begin{definition}
The \emph{number of rounds} of a public coin protocol is the maximum depth of the protocol tree $\pi_r$ over all 
choices of the public randomness. 
\end{definition}

\section{Proof of Claim 3 (from \cite{BravermanInteractive11})}

\begin{proof}
Recall that $\Div\mu\nu = \sum_{x\in\cU} \mu(x) \log \frac{\mu(x)}{\nu(x)}$. Denote by $\cN=\{x:~\mu(x)<\nu(x)\}$ 
-- the terms that contribute a negative amount to $\Div\mu\nu$.
First we observe that for all $0<x<1$, $x\log x>-1$, and thus
$$
\sum_{x\in\cN} \mu(x) \log \frac{\mu(x)}{\nu(x)} = 
\sum_{x\in\cN} \nu(x)\cdot \frac{\mu(x)}{\nu(x)} \log \frac{\mu(x)}{\nu(x)} 
\ge \sum_{x\in\cN} \nu(x)\cdot (-1) > -1. 
$$
Denote by $\cL = \left\{x:~2^{(I+1)/\ve}\cdot \nu(x) <\mu(x)\right\}$; we need to show that $\mu(\cL)<\ve$. 
For each $x\in\cL$ we have $\log \frac{\mu(x)}{\nu(x)}>(I+1)/\ve$. 
Thus
$$
I \ge \Div{\mu}{\nu} \ge \sum_{x\in\cL} \mu(x) \log \frac{\mu(x)}{\nu(x)} + \sum_{x\in\cN} \mu(x) \log \frac{\mu(x)}{\nu(x)} >
\mu(\cL)\cdot (I+1)/\ve -1,
$$
implying $\mu(\cL)<\ve$. 
\end{proof}

\section{Proof of Claim 4} 
\begin{proof}

For any $x_i\in\cU$, 

\begin{eqnarray}
\mu_1(x_i)= \Pr(a = x_i \; | \; \suc) \leq \frac{\mu(x_i)2^{-50{(I+1)}}}{\Pr[\suc]} \leq \frac{\mu(x_i)}{0.9} = (1 + 1/9)\mu(x_i)
\end{eqnarray}

\noindent where the last inequality follows from Proposition \ref{s_bound}. Hence,
$|\mu_1 - \mu| = $

\begin{eqnarray}
2\bigg{(}\sum_{x_i : \mu_1(x_i) \geq \mu(x_i)} \mu_1(x_i) - \mu(x_i)\bigg{)} \leq 2\bigg{(} \sum_{x_i : \mu_1(x_i) 
\geq \mu(x_i)}(1 + 1/9)\mu(x_i) - \mu(x_i) \bigg{)} \leq \frac{2}{9} \nonumber
\end{eqnarray}

\noindent This proves claim (3) of the lemma.
\end{proof}

\section{Proof of Lemma \ref{lem:GT}: The discrepancy of the Greater-Than function}

We consider the Greater-Than function on $n$-bit strings.
We start by defining the ``hard" distribution $\mu$. A pair $(x,y)$ is sampled as follows:
\begin{enumerate}
\item 
Sample an index $k\in\{1,\ldots,n\}$ uniformly at random.
\item 
Sample $z_1,\ldots,z_{k-1}$, $w$, $x_{k+1},\ldots,x_n$, $y_{k+1},\ldots,y_n$ --- uniformly random bits. 
\item 
Let $x=z_1,\ldots,z_{k-1},w,x_{k+1},\ldots,x_n$, $y=z_1,\ldots,z_{k-1},\overline{w},y_{k+1},\ldots,y_n$.
\end{enumerate}
Denote this distribution by $\mu_n$. Let $GT_n(x,y)=1$ iff $x>y$. 
We will prove the following Lemma:

\begin{lemma}
\label{lem:GTmain}
The discrepancy of $GT_n$ with respect to $\mu_n$ satisfies $$Disc_{\mu_n}(GT_n)<\frac{20}{\sqrt{n}}.$$ 
\end{lemma}

In fact, to facilitate an inductive proof, we will show a slightly stronger statement:

\begin{lemma}
\label{lem:GTaux}
Let $R=S\times T$ be a rectangle in $\{0,1\}^n\times\{0,1\}^n$. Let $s:=|S|/2^n$ and $t:=|T|/2^n$ be the uniform 
size of $S$ and $T$ respectively. Then 
$$
Disc_{\mu_n}(GT_n,R)<\frac{20 \sqrt{s t}}{\sqrt{n}}.
$$
\end{lemma}

Note that Lemma~\ref{lem:GTaux} immediately implies Lemma~\ref{lem:GTmain}.

\begin{proof} 
We prove Lemma~\ref{lem:GTaux} by induction on $n$. The statement is trivially true for $n=1$. 
Assume the statement is true for $n-1$. Our goal is to prove it for $n$. Let $R=S\times T$ be 
any rectangle in $\{0,1\}^n\times\{0,1\}^n$.
By a slight abuse of notation we write:
$$
Disc_{\mu_n}(GT_n,R) =  \Pr_{\mu_n}[f(x,y) = 1 \wedge (x,y)\in R] - \Pr_{\mu_n}[f(x,y) = 0 \wedge (x,y)\in R],
$$
and prove an upper bound on this quantity (without $|\,\cdot\,|$). The matching upper bound on $-Disc_{\mu_n}(GT_n,R)$ 
follows by an identical argument. 

 Let $s:=|S|/2^n$ and $t:=|T|/2^n$. Denote by $S_0$ 
the set of strings in $S$ that begin with a $0$, and $S_1:=S\setminus S_0$. Similarly, define $T_0$ and $T_1$. 
Further, let $p:=|S_0|/|S|$ and $q:=|T_0|/|T|$. 

Note that restricted to $S_0\times T_0$, $\mu_{n}$ is the same distribution as $\mu_{n-1}$, scaled by a factor 
of $\frac{n-1}{2n}$. Moreover, $s_0:=|S_0|/2^{n-1}=p s  2^{n}/2^{n-1} = 2ps$. Similarly, $s_1:=|S_1|/2^{n-1}=2(1-p)s$, 
$t_0:=|T_0|/2^{n-1} = 2qt$, $t_1:=|T_1|/2^{n-1}=2(1-q)t$. Putting these pieces together, and applying the inductive 
hypothesis we get:
\begin{multline} \label{eq:gt1}
Disc_{\mu_n}(GT_n,S\times T) = Disc_{\mu_n}(GT_n,S_0\times T_0) +Disc_{\mu_n}(GT_n,S_1\times T_1) + \\  Disc_{\mu_n}(GT_n,S_1\times T_0)
+Disc_{\mu_n}(GT_n,S_0\times T_1)    = \\ 
\frac{n-1}{2 n}\cdot Disc_{\mu_{n-1}}(GT_{n-1},S_0\times T_0) + \frac{n-1}{2 n}\cdot Disc_{\mu_{n-1}}(GT_{n-1},S_1\times T_1) + \\
\frac{2}{n}(1-p)s q t - \frac{2}{n}p s (1-q) t  < \\
\frac{n-1}{2n}\cdot \frac{20 \sqrt{s_0 t_0}}{\sqrt{n-1}} + \frac{n-1}{2n}\cdot \frac{20 \sqrt{s_1 t_1}}{\sqrt{n-1}} + \frac{2}{n} (q-p) st = \\
\frac{1}{\sqrt{n}} \left( \sqrt{\frac{n-1}{n}} \cdot \left( {20 \sqrt{pq}\cdot\sqrt{st}}+{20 \sqrt{(1-p)(1-q)}\cdot\sqrt{st}}\right) +\frac{2}{\sqrt{n}} (q-p) st \right).
\end{multline}
If $q-p<0$, we continue \eqref{eq:gt1} as follows:
\begin{multline*}
RHS \le \frac{1}{\sqrt{n}} \left( \sqrt{\frac{n-1}{n}} \cdot \left( {20 \sqrt{pq}\cdot\sqrt{st}}+{20 \sqrt{(1-p)(1-q)}\cdot\sqrt{st}}\right)\right)\le \\
\frac{20\sqrt{st}}{\sqrt{n}} \cdot \left(\sqrt{pq}+\sqrt{(1-p)(1-q)}\right) \le \frac{20\sqrt{st}}{\sqrt{n}},
\end{multline*}
where the last inequality follows from simple calculations. 

On the other hand, in the more difficult case when  $q-p\ge 0$, we use the fact the $st\le 1$ to continue \eqref{eq:gt1} as follows:
\begin{multline}
\label{eq:gt2}
RHS \le \\ \frac{1}{\sqrt{n}} \left( \sqrt{\frac{n-1}{n}} \cdot \left( {20 \sqrt{pq}\cdot\sqrt{st}}+{20 \sqrt{(1-p)(1-q)}\cdot\sqrt{st}}\right) +\frac{2}{\sqrt{n}} (q-p) \sqrt{st} \right) = \\
\frac{20\sqrt{st}}{\sqrt{n}} \left( \sqrt{\frac{n-1}{n}} \cdot \left( { \sqrt{pq}}+{ \sqrt{(1-p)(1-q)}}\right) +\frac{1}{10 \sqrt{n}} (q-p) \right) 
\end{multline}
Next, we use the readily verifiable facts that $\sqrt{\frac{n-1}{n}} < 1-\frac{1}{2n}$ and that ${ \sqrt{pq}}+{ \sqrt{(1-p)(1-q)}} \le 1-(q-p)^2/4$,
to continue \eqref{eq:gt2} as follows:
\begin{multline}
\label{eq:gt3}
RHS \le \frac{20\sqrt{st}}{\sqrt{n}} \left( \left(1-\frac{1}{2n}\right) \cdot \left( 1-(q-p)^2/4 \right) +\frac{1}{10 \sqrt{n}} (q-p) \right) \le\\
 \frac{20\sqrt{st}}{\sqrt{n}} \left( \left(1-\frac{1}{4n}-(q-p)^2/8 \right) +\frac{1}{10 \sqrt{n}} (q-p) \right) =
\\
 \frac{20\sqrt{st}}{\sqrt{n}} \left( 1-\left(\frac{1/(2n)+(q-p)^2/4}{2} \right) +\frac{1}{10 \sqrt{n}} (q-p) \right)  \le \\
  \frac{20\sqrt{st}}{\sqrt{n}} \left( 1-\sqrt{\frac{1}{2n}\cdot\frac{(q-p)^2}{4}} +\frac{1}{10 \sqrt{n}} (q-p) \right)=\\
    \frac{20\sqrt{st}}{\sqrt{n}} \left( 1-\frac{1}{\sqrt{8} \sqrt{n}} (q-p) +\frac{1}{10 \sqrt{n}} (q-p) \right) \le  \frac{20\sqrt{st}}{\sqrt{n}},
\end{multline}
where the third-to-last inequality follows from the fact that for all $0\le a,b\le 1$, $(1-a)(1-b)\le 1-a/2-b/2$, and the second-to-last one in an application 
of the AM-GM inequality. 
\end{proof}

\newpage
\section{Sampling protocol from Lemma \ref{lem_sample}}

\begin{figure}[h!tb] 
\begin{tabular}{|l|}
\hline
\begin{minipage}{\algwidth}
\vspace{1ex}
\begin{center}
\textbf{Information-cost sampling protocol $\Pi_1$} 
\end{center}
\vspace{0.5ex}
\end{minipage}\\
\hline
\begin{minipage}{\algwidth}
\vspace{1ex}
\begin{enumerate}
    \item Alice  computes the set $\cA$. Bob  computes the set $\cB$. 
    \item If $\cA=\emptyset$, the protocol fails, otherwise Alice finds the first element $a \in \cA$, and sends $a$ to Bob.
    \item \label{find} Bob checks if $a \in \cB$. If not, the protocol fails.
    \item Alice and Bob output $a$ (``success").  
\end{enumerate}
\vspace{0.3ex}
\end{minipage}\\
\hline
\end{tabular}
\caption{The sampling protocol $\Pi_1$ from Lemma~\ref{lem_sample} }\label{figure:pi1}
\end{figure}

\end{document}